\documentclass{article}
\usepackage{arxiv}
\usepackage[caption=false]{subfig}
\usepackage{graphics}
\usepackage{graphicx} 
\usepackage{amsmath}
\usepackage{amsfonts}
\usepackage{amsthm}
\usepackage{physics}
\usepackage{amsmath}
\usepackage{algorithm,float}
\usepackage{mathtools}
\usepackage{algpseudocode}
\usepackage{comment}
\usepackage{xurl}
\usepackage[hidelinks]{hyperref}

\newcommand{\sfunction}[1]{\textsf{\textsc{#1}}}
\usepackage{arxiv}
\usepackage[caption=false]{subfig}
\usepackage{graphics}

\usepackage{amsmath}
\usepackage{import}
\usepackage{pdfpages}
\usepackage{amsthm}
\usepackage{amssymb}
\usepackage{ulem}
\usepackage{mathtools}
\usepackage{algorithm,float}
\usepackage[hidelinks]{hyperref}
\usepackage{comment}
\usepackage{url}
\usepackage{xcolor}
\usepackage{soul}

\newtheorem{theorem}{Theorem}
\newtheorem{lemma}{Lemma}

\algrenewcommand\algorithmicforall{\textbf{\PW{for each}}}
\providecommand{\keywords}[1]
{
  \small	
  \textbf{\textit{Keywords ---}} #1
}

\usepackage{color}

\title{Kernel-based learning with guarantees for~multi-agent applications*}

\author{Krzysztof Kowalczyk\inst{1}\orcidID{0009-0003-3972-5738} \and
Pawe\l{} Wachel\inst{1}\orcidID{0000-0002-7353-2310} \and
Cristian R. Rojas\inst{2}\orcidID{0000-0003-0355-2663}}
\date{15 April 2024}

\begin{document}

\author{	Krzysztof Kowalczyk\\
	Department of Control Systems and Mechatronics\\
	Wroc\l{}aw University of Science and Technology\\
	Wroc\l{}aw, Poland \\
	\texttt{krzysztof.kowalczyk@pwr.edu.pl} \\
	\And
Pawe\l{} Wachel\\
	Department of Control Systems and Mechatronics\\
	Wroc\l{}aw University of Science and Technology\\
	Wroc\l{}aw, Poland \\
	\texttt{pawel.wachel@pwr.edu.pl} \\
	\And
Cristian R. Rojas\\
	School of Electrical Engineering and Computer Science\\
	KTH Royal Institute of Technology\\
	Stockholm, Sweden \\
	\texttt{crro@kth.se} 
}

\maketitle              
\begin{abstract}
This paper addresses a kernel-based learning problem for a network of agents locally observing a latent multidimensional, nonlinear phenomenon in a noisy environment. 
We propose a learning algorithm that requires only mild \textit{a priori} knowledge about the phenomenon under investigation and delivers a model with corresponding non-asymptotic high probability error bounds.

Both non-asymptotic analysis of the method and numerical simulation results are presented and discussed in the paper.

\end{abstract}
\let\thefootnote\relax\footnotetext{*A preprint submitted to ICCS2024.}
\section{Introduction}

A multi-agent system is a network of autonomous entities called agents that share information and collaborate to solve tasks usually beyond an individual agent's scope~\cite{jain2010innovations}. This broad description fits well in the recent research trends like cloud computing \cite{sim2011agent}, or Industry 4.0 \cite{sakurada2020multi}, and allows multi-agent systems to find applications in many other fields. In robotics, in scenarios including groups of mobile robots or swarms of drones, it is necessary to avoid collisions or obstacles and to navigate collaboratively~\cite{rasheed2022review}. The agent-based approach is also used for controlling smart grids, \textit{i.e.}, efficient and robust power systems~\cite{mahela2020comprehensive}. We can also find numerous other examples, like analyzing the traffic flow \cite{malecki2023multi} or modelling purchasing decisions \cite{jkedrzejewski2022purchasing}.

Inspired by these multidisciplinary applications, we formally discuss the general problem of distributed learning, with a particular focus on the modelling of nonlinearities under limited information, \textit{cf.} \cite{lagosz2021identification}. In the considered scenario, every agent (node) locally observes the outcome of some unknown global phenomenon of interest. Although the agents aim to provide a non-local comprehensive model of the phenomenon, this may be not possible for individual nodes due to the limited range of their own observations. Thus, collaboration is necessary. Nonetheless, we assume that the agents cannot communicate freely throughout the entire network, but a single agent can only interact with a narrow group of its neighbourhood nodes (\textit{cf.} Fig. \ref{fig:Drones}).

One can find numerous approaches related to this problem in the literature, among which Kalman-based filtering \cite{cattivelli2010distributed}, diffusion \cite{lopes2008diffusion}, and consensus \cite{bertrand2011consensus} techniques can be distinguished; see \textit{e.g.} \cite{modalavalasa2021review} for a more extensive discussion. While our approach is motivated by the abovementioned methods, we introduce, however, a few substantial modifications. In particular, regarding the investigated nonlinear phenomenon, we require only limited \textit{a priori} knowledge, usually insufficient for many parametric estimation techniques proposed so far. We use kernel regression for efficient non-parametric modelling and provide corresponding error-bound guarantees that hold for a finite number of samples. 
\begin{figure}[ht]
    \centering
\includegraphics[width=0.8\textwidth]{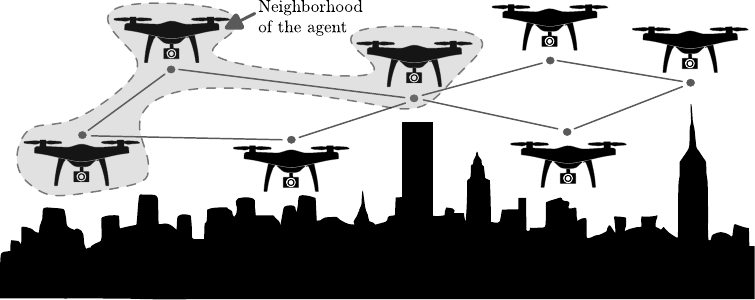}
    \caption{A network of distributed agents with highlighted neighbourhood of a selected node.}
    \label{fig:Drones}
\end{figure}
The algorithm proposed in this paper is an extension of the method introduced in \cite{WACHEL2023100912}. Our main contributions are as follows:
\begin{itemize}
    \item We propose a new algorithm for distributed learning in the multi-agent scenario, where a group of agents locally observes a nonlinear phenomenon.
    \item We formally investigate the algorithm, allowing multidimensionality of explanatory data (measurements) and derive corresponding confidence bounds.
    \item We show that the resulting error bounds are independent of the dimension of the explanatory measurements.
    \item The formal analysis of the algorithm is performed under mild assumptions on the phenomenon (essentially only Lipschitz continuity is required). In particular, we do not assume any parametric structure of the investigated phenomenon.
    \item The data exchange protocol allows for obtaining results close to the corresponding centralized version of the problem.
    \item The algorithm has a simple and direct construction (internal optimization routines and other numerically intensive procedures are not involved).
\end{itemize}
    
The paper is organized as follows. Section~$2$ describes the considered problem and contains the required formal assumptions. In Sections~$3$ and~$4$, the main algorithm is presented, and its estimation guarantees are derived. In Section~$5$  the results of simulation experiments are presented and discussed, whereas concluding remarks are given in Section~$6$. Finally, technical lemmas and proofs are collected in the appendix.  

\section{Problem formulation}
We investigate a problem of distributed learning, where a group of agents observes an unknown phenomenon in a noisy environment and aims to provide noise-free estimations with high probability guarantees for a given region of interest. 
    
    As is often the case in practice, the agents operate in restricted subspaces; see Fig. \ref{fig:Domains}. In consequence, their observations/measurements have merely a local character, and the learning outcomes outside their local domains may be unsatisfactory due to poor prior knowledge and/or (significant) sensing inaccuracies. Hence, to provide better-quality estimation results, it is often required to introduce collaboration between the agents. In Fig. \ref{fig:Domains}, we provide an example of a nonlinear phenomenon and the local scopes (operating areas) of agents that attempt to model it. As it can be noticed, due to the scattering of agents and their relatively narrow measuring domains, it would be impossible for a single agent to provide an even rough approximation of the whole phenomenon. 
\begin{figure}[ht]
    \centering
\includegraphics[width=0.6\textwidth]{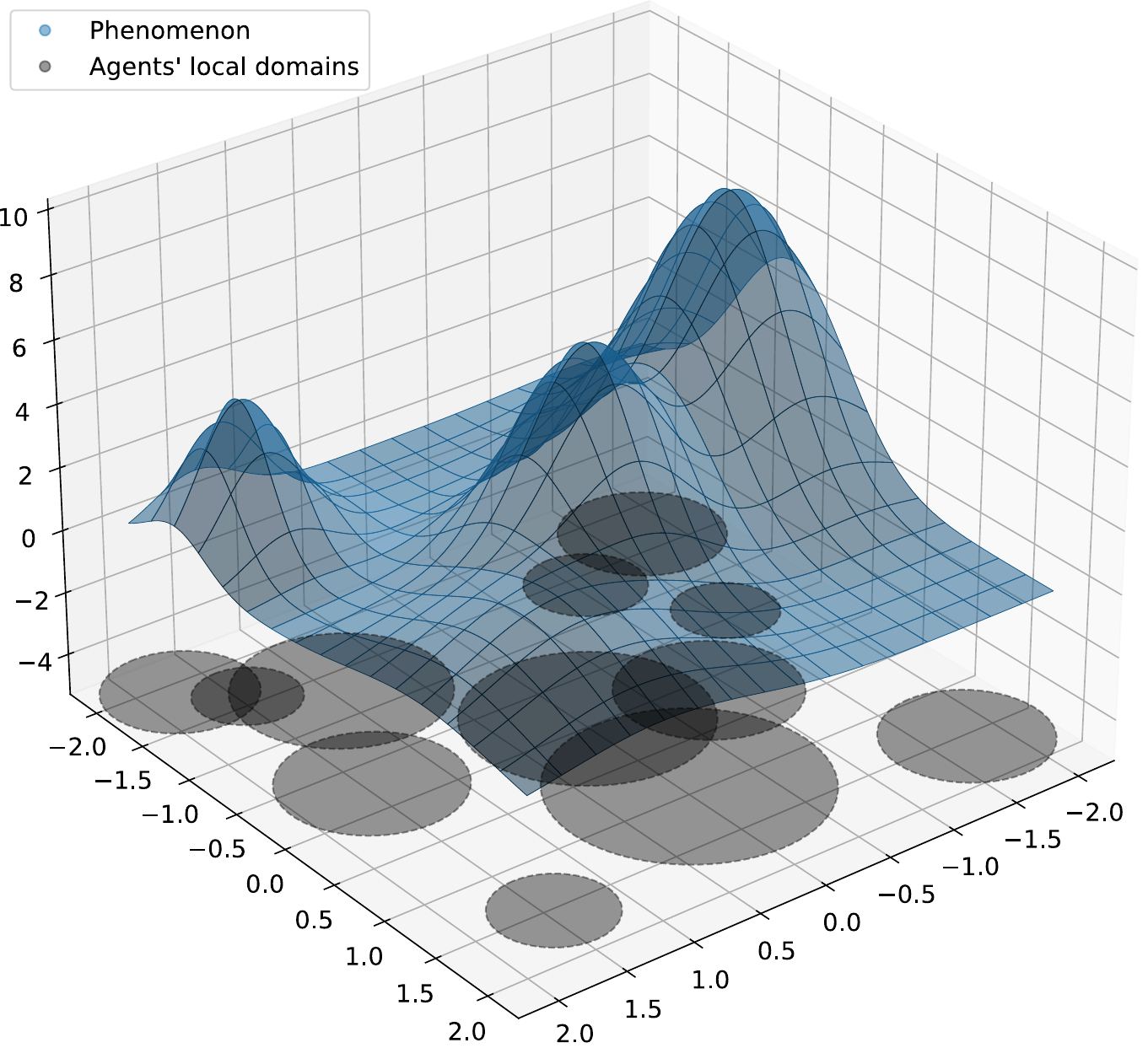}
    \caption{Example of a nonlinear phenomenon (blue) with marked agents' local domains of observations/measurements (grey).}
    \label{fig:Domains}
\end{figure}
\begin{subsection}{Network setup}
We consider a set of $M$ agents and model their cooperation via a connected and undirected graph $\mathcal{G}=(\mathcal{M},\mathcal{E})$ with $\mathcal{M}=\{1,2,\dots,M\}$ nodes (corresponding to the agents) and a set of unweighted edges $\mathcal{E}$. To reflect possible communication restrictions and to reduce the communication burden, we assume that two nodes $i,j\in\mathcal{M}$ can exchange information if and only if they are directly connected, \textit{i.e.}, if $\{i,j\}\in\mathcal{E}$. Thus, we define the neighbourhood of a node $i\in\mathcal{M}$ as the set $\mathcal{N}_i=\{j\colon \{i,j\}\in\mathcal{E}\}$. 
\end{subsection}

\begin{subsection}{Phenomenon and observations}
In the considered setup, at every time step $t\in\mathbb{N}$ , every agent $k\in\mathcal{M}$ obtains an explanatory data point $\xi_{k,t}\in\mathbb{R}^p$, for some fixed $p\in\mathbb{N}$, and observes a noisy outcome $y_{k,t}$ of the latent nonlinear phenomenon modelled by an unknown nonlinear mapping $m:\mathcal{D}\subset\mathbb{R}^p\rightarrow\mathbb{R}^d$,
\begin{equation}
    y_{k,t}=m(\xi_{k,t})+\eta_{k,t}, \quad k\in\mathcal{M},\quad t\in\mathbb{N},
\end{equation}
where $\eta_{k,t}$ denotes an additive noise. 
\end{subsection}

This paper aims to provide a distributed inference of $m$ under mild \textit{a priori} knowledge about its structure. Hence, the following assumptions regarding the observed phenomenon and the additive noise have a general form. In the sequel, for simplicity of notation, we will use the symbol $a_{1:m}$ as a short for a sequence $a_1,\dots,a_m$.

    \textit{Assumption} $1$. The latent phenomenon of interest, $m\colon \mathcal{D} \subset \mathbb{R}^p\rightarrow \mathbb{R}^d$, is a Lipschitz continuous mapping, \textit{i.e.}
    \begin{equation}
        \norm{m(\xi)-m(\xi')}_2\leq L\norm{ \xi-\xi'}_2,\quad \forall\;\xi,\xi'\in\mathcal{D},
    \end{equation}
    for a known constant $0 \leq L<\infty$.
    
\textit{Assumption} $2$. The explanatory sequence $\{\xi_{t}\in\mathbb{R}^p \colon t\in\mathbb{N}\}$ is an arbitrary stochastic process. 

\textit{Assumption} $3$. The disturbance $\{\eta_t\in\mathbb{R}^d \colon\,t\in\mathbb{N}\}$ is a sub-Gaussian stochastic process, that is, there exists some $\sigma>0$ such that, for every $\gamma_t \in\mathbb{R}^d$ (possibly a function of $\xi_t$), and every $t\in\mathbb{N}$,
\begin{equation}\label{eq:A_eta}
\mathbb{E}\{ \exp ( \gamma_t^\top \eta_t)|\eta_{1:t-1}, \xi_{1:t}\} \leq \exp
\left( \frac{\gamma_t^\top\gamma_t\sigma ^{2}}{2}\right).
\end{equation}

The above requirements have a somewhat general character and are inspired by the real-world properties of many technical processes. Informally, Assumption~$1$ allows, in particular, any nonlinear function with a limited rate of increase (or decrease), and Assumption~$3$ admits any bounded \textit{i.i.d.}~disturbances with zero mean, independent of the explanatory data.

\section{Local agents' modelling}
To construct the proposed learning technique, we begin from a single-agent perspective. Given a fixed time instant $t$ and a set of local data measurements, we define for agent $k\in \mathcal{M}$ the following kernel regression estimator:
\begin{equation}\label{eq:NW1}
  \begin{aligned}
\hat{\mu}_{k,t}(x):=&\sum_{n=1}^{t}\frac{K_h(x,\xi_{k,n})}{\kappa_{k,t}(x)}y_{k,n}=:\frac{\psi_{k,t}(x)}{\kappa_{k,t}(x)},\\
\quad\kappa _{k,t}(x):=&\kappa_{k,t}(x,h)=\sum_{n=1}^{t}K_h(x,\xi_{k,n}),
 \end{aligned}
\end{equation}
with $K_h(x,\xi):=K(\norm{x-\xi}_2/h)$, and where $K$, $h$ are the kernel function and the bandwidth parameter, respectively. As can be easily noticed, $\hat{\mu}_{k,t}$ is a local average of output measurements with weights controlled by kernel $K$. To ensure appropriate statistical properties of $\hat{\mu}_{k,t}$, we make the following assumption:

\textit{Assumption} $4$. The kernel $K\colon \mathbb{R}\rightarrow\mathbb{R}$ is such that $0 \leq K(v)\leq1$ for all $v\in\mathbb{R}$. Also, $K(v)=0$ for all $|v|>1$.

We are now about to develop the main technical result, which is the basis for the network estimation algorithm introduced in the sequel (\textit{cf.} \cite{WACHEL2023100912}).
\begin{lemma}\label{L:local}
Let Assumptions 1--4 be in force. Consider the estimator $\hat{\mu}_{k,t}\in\mathbb{R}^d$ and fix a bandwidth parameter $h$. Let $x \in \mathcal{D}\subset\mathbb{R}^p$ be  fixed or in general a measurable function of $\eta_{k,1:t-1}, \xi_{k,1:t}$ (e.g., $x=\xi_{k,t}$). Then, for every $0 <\delta <1$, with probability at least $1-\delta $, if $\kappa _{k,t}(x)\neq 0$,
\begin{align}
\lVert \hat{\mu}_{k,t}(x) -m(x)
\rVert_2 \leq \beta_{k,t}(x),\\  \label{b1}
\quad\text{ where}\quad
\beta_{k,t}(x) := L h+2\sigma \frac{\alpha _{k,t}(x ,\delta )}
{\kappa _{k,t}(x)}
\end{align}
and 
\begin{equation} \label{b2}
\alpha _{k,t}( x ,\delta ) :=
\begin{cases} 
\sqrt{\log(\delta^{-1}2^{d/2})}, & \textnormal{for } \kappa _{k,t}(x) \leq 1 \vspace{3pt}\\ 
\sqrt{\kappa_{k,t}(x) \log \left( \delta^{-1} \bigl( 1 + \kappa_{k,t}(x)\bigr)^{d/2}\right)}, & \textnormal{for } \kappa _{k,t}(x) > 1.
\end{cases}
\end{equation}
\end{lemma}
\begin{proof}
See the Appendix. 
\end{proof}
In Lemma 1, we provide error bounds for local (single-agent) estimates that hold with probability $1-\delta$,  where $\delta$ is a user-defined failure probability. The Lipschitz constant $L$ and the noise proxy variance $\sigma$ are, however, required to be known (in practice, at least upper bounds on these quantities are needed).
Observe that, since the denominator $\kappa_{k,t}(x)$ depends on the measurements, it will increase with the number of data samples, which will result in tighter error bounds. 

Due to the fact, that the dimensionality of the output influences the bounds, for higher $d$'s, it may be worth considering techniques of MIMO system decompositions as $\textit{e.g.}$ \cite{wachel2023learning}.
\section{Distributed modelling -- data aggregation}
Having a single-agent estimator, we are now ready to introduce a distributed modelling procedure.

According to the considered approach, every agent $k$ spreads its local estimations by broadcasting tuples of essential data $T_{k,t}(x)=(\psi_{k,t}(x),\allowbreak\kappa_{k,t}(x),\allowbreak x)$, which contains locally computed numerator, denominator and the estimation point, to its neighbourhood $\mathcal{N}_k$. The acquired tuples are then stored in set $\mathbb{T}_k$. To avoid data repetition in a container of tuples, only a single tuple from a single agent and fixed estimation point $x$ is included in  $\mathbb{T}_k$, \textit{i.e.}, the newer (incoming) tuples overwrite the older ones. 

More precisely, in the following pseudo-code (Algorithm 1)   we propose a data exchange and aggregation procedure for every agent in the network. The proposed mechanism allows for the online generation/update of tuple containers.
\begin{algorithm}[H]
\caption{Data exchange and aggregation\Comment{Agent $k$}} \label{Alg:1}

\begin{algorithmic}[1]
\State \textbf{input:} $\mathcal{X}$ \Comment{Estimation points}
\For{$t=1,2,\dots$}
    \State $\sfunction{Get }(\xi_{k,t},y_{k,t})$ \Comment{Get local measurement}
    \If {acquired\_new\_tuple}
        \State $\sfunction{Update }\mathbb{T}_k$
    \EndIf
    \If {send\_local\_data}
    \State $\sfunction{Select }x\in\mathcal{X}$ \Comment{Select an estimation point}
    \State $\sfunction{Evaluate }\psi_{k,t}(x),\kappa_{k,t}(x)$ 
    \State $T_{k,t}(x)\leftarrow(\psi_{k,t}(x),\kappa_{k,t}(x),x)$
    \EndIf
    \If {send\_acquired\_data}
        \State $\sfunction{Select }T_i(x)\in\mathbb{T}_k$
    \EndIf
    \State $\sfunction{Broadcast } \text{selected tuple}$ \Comment{Send data to the neighbors}
    \State \textbf{end}
    
\EndFor
\State \textbf{end}
\end{algorithmic} 
\end{algorithm}
The above algorithm requires a few comments. We assume that all the agents work on the same set $\mathcal{X}$ (\textit{i.e.}, $x\in\mathcal{X}$) and they can freely share their data. We do not specify here when the agents should transfer their local data and when they acquire information from their neighbourhoods. Currently, we leave this open for the user, by setting the flags \textit{send\_local\_data} and \textit{send\_acquired\_data}  (in the experiments these flags were set randomly).

Following the data exchange and aggregation routine proposed in Algorithm~1, every agent builds a tuple set $\mathbb{T}_k$ that will be used next to construct a model of $m(\cdot)$. For every agent $k$ with $\mathbb{T}_k$, we define an estimator that combines all the acquired data as follows:
\begin{equation*}\label{globalest}
\hat{m}_{k,t}(x)=\frac{\sum_{i=1}^{M} \psi_i(x)}{\sum_{i=1}^{M}\kappa_i(x)}=\frac{\Psi_{k,t(x)}}{\mathcal{K}_{k,t}(x)},  \quad \psi_i(x),\kappa_i(x)\in T_{i}(x)\in\mathbb{T}_{k}.
\end{equation*}

For the estimator in \eqref{globalest}, we provide non-asymptotic error bounds in Theorem~\ref{th1} below.
\begin{theorem}\label{th1}
Let Assumptions 1--4 be in force. Consider any agent $k\in\mathcal{M}$ with data exchange and
aggregation procedure as in Algorithm 1 and estimate $\hat{m}_{k,t}$. Then, for $x\in\mathcal{X}$ and any $0<\delta<1$, with probability $1-\delta$,
\begin{align}
\lVert \hat{m}_{k,t}(x) -m(x)
\rVert_2 \leq \beta_{k,t}(x),
\end{align}
where $\beta_{k,t}(x)$ is given by eqns.~\eqref{b1} and~\eqref{b2}.
\end{theorem}

\begin{proof}
By definition,
\begin{align}
\hat{m}_{k,t}=\frac{\sum_{i=1}^{M} \psi_i(x)}{\sum_{i=1}^{M}\kappa_i(x)}=&\frac{\sum_{i=1}^{M}\sum_{n=1}^{t_i} K_h(x,\xi_{i,n})y_{i,n}}{\sum_{i=1}^{M}\sum_{n=1}^{t_i} K_h(x,\xi_{i,n})},
\end{align}
where $t_i$ denotes the time step at which $\psi_i$ and $\kappa_i$ were calculated. Now, we introduce a merged index $q$ that takes values from $1$ to $\tau=\sum_{i=1}^{M}t_i$ and mappings $i_q$ and $n_q$, that transfer a single $q$ back to the original $i$ and $n$, respectively. Thus,
\begin{align}
    \frac{\sum_{i=1}^{M}\sum_{n=1}^{t_i} K_h(x,\xi_{i,n})y_{i,n}}{\sum_{i=1}^{M}\sum_{n=1}^{t_i} K_h(x,\xi_{i,n})}=&\frac{\sum_{q=1}^{\tau}K_h(x,\xi_{i_q,n_q})y_{i_q,n_q}}{\sum_{q=1}^{\tau}K_h(x,\xi_{i_q,n_q})}\\
    =&\frac{\sum_{q=1}^{\tau}K_h(x,\xi_{i_q,n_q})}{\kappa_{\tau}(x)}y_{i_q,n_q}\\
    =&\hat{\mu}_{\tau}(x).
\end{align}
This can be interpreted as the local estimator of an agent, that directly acquired all $\tau$ observations. Hence, we can apply the error bound from Lemma 1, which completes the proof. 
\end{proof}
As we have shown, Theorem 1 can be proven by reinterpreting Lemma~1 since the final estimate combines the acquired numerators and denominators, and is, in fact, the same as the estimate calculated from raw data transferred to a single agent. This is however possible only if all the agents operate with the same upper bound of the noise proxy variance $\sigma$.

\section{Numerical experiments}
In this section, we illustrate the main concept of the proposed approach\footnote{The Python code to obtain the numerical results is available at \url{https://github.com/kkowalc/Kernel-based-learning-with-guarantees-for-multi-agent-applications}.}. To this end, we use a network of 25 agents with randomly selected topology, as shown in Fig. \ref{fig:Topology}. 
\begin{figure}[h!]
    \centering
\includegraphics[width=0.6\textwidth]{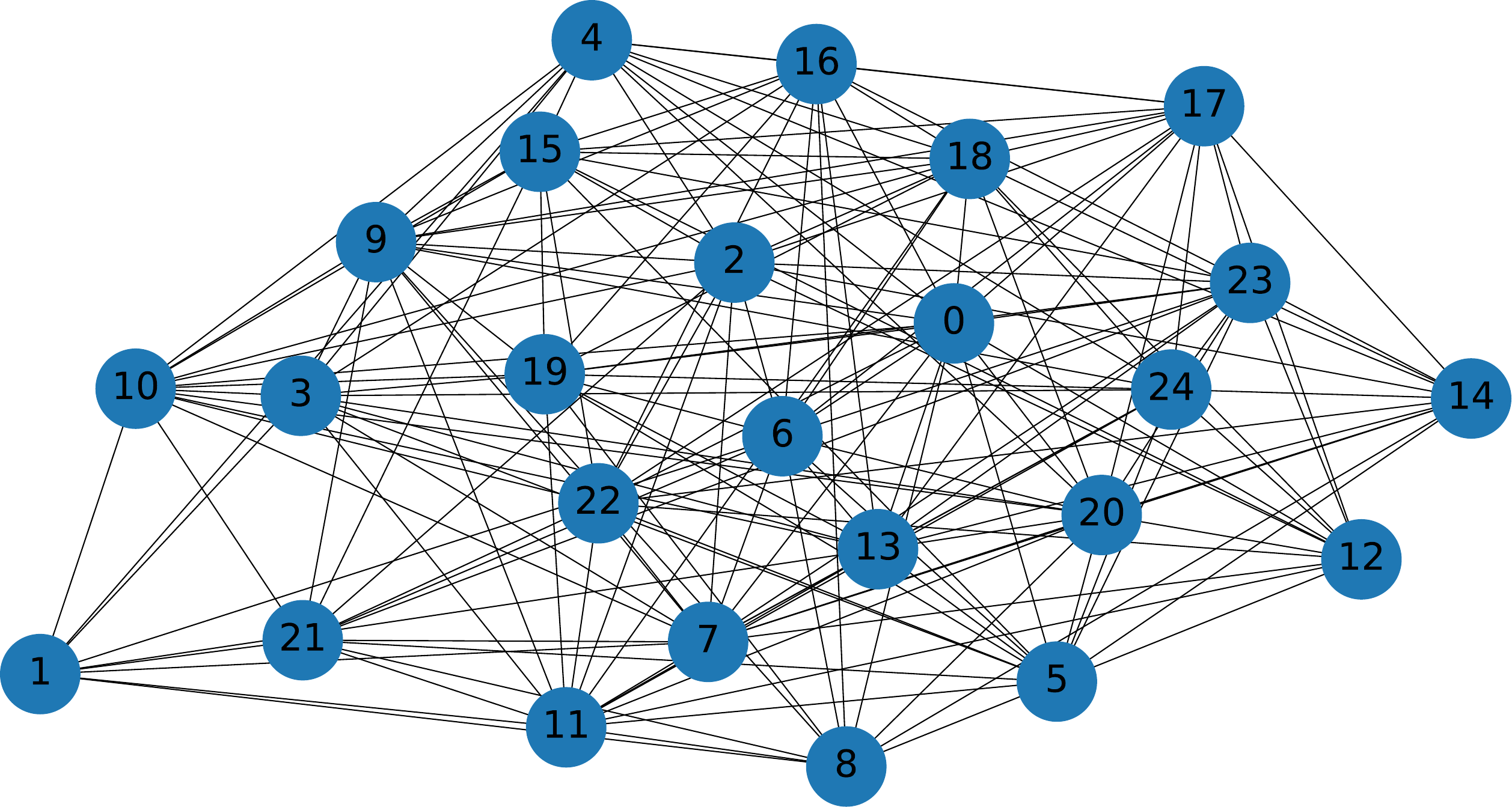}
    \caption{Random topology network with 25 nodes. }
    \label{fig:Topology}
\end{figure}

In the experiments we consider a nonlinearity $m\colon \mathbb{R}^2\rightarrow\mathbb{R}$ being a mixture of three Gaussian surfaces $\mathcal{N}([0,0], 0.5\mathbb{I})$, $\mathcal{N}([1,2], 0.55\mathbb{I})$, $\mathcal{N}([2,-2], 0.7\mathbb{I})$. The output noise sequences $\eta_{k,t}$ for every agent $k$ are sampled from a normal distribution $\mathcal{N}(0,0.05)$. We assume that the total region of interest $\mathcal{D}$ is a set $[-2,2]\times[-2,2]$ and the estimation grid $\mathcal{X}$ is evenly spaced with a step $0.25$. The explanatory data $\xi_{k,t}$ is generated from a normal distribution $\mathcal{N}(\mu_{\xi_k},\sigma_{\xi_k})$. Both $\{\xi_{k,t}\}$ and $\{\eta_{k,t}\}$ are mutually independent. For simplicity of calculations and clarity of presentation,  the parameters $\mu_{\xi,k}$ and $\sigma_{\xi,k}$ are selected to ensure that $\mathcal{D}\subset\mathcal{D}_1\cup\mathcal{D}_2\cup\dots\cup\mathcal{D}_M$; otherwise, it would be necessary to propagate the bounds with the Lipschitz constant for the regions where measurements could not be obtained. The required parameters $L$  and $\delta$ are set to $0.3$ and $0.001$, respectively. 

\begin{figure}[h]
    \centering
\includegraphics[width=0.6\textwidth]{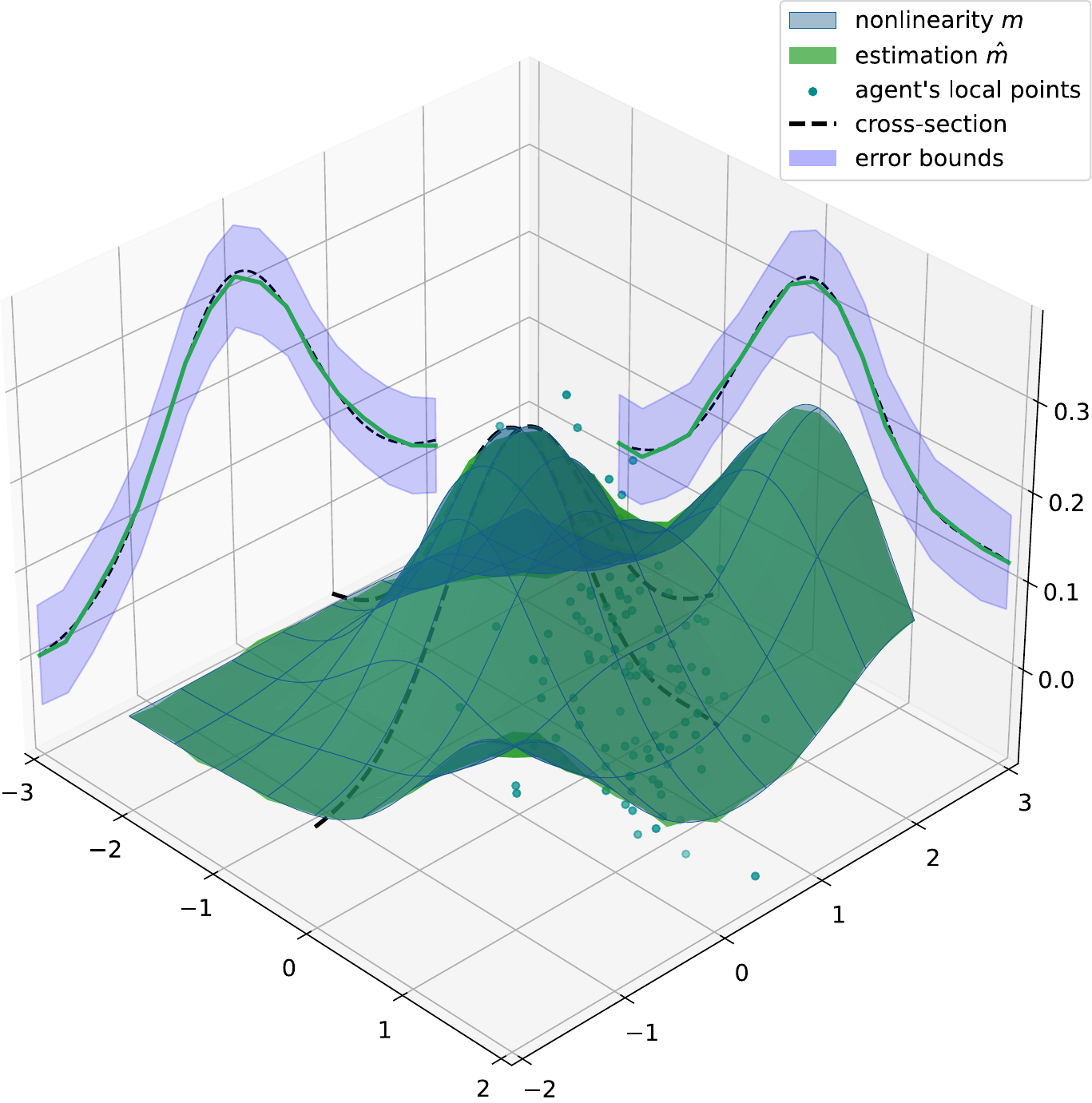}
    \caption{A model of the estimated phenomenon provided by agent $k=13$ after $t=5000$ time steps, with cross-sections presenting its confidence bounds.}
    \label{fig:NS1}
\end{figure}
In Fig.~\ref{fig:NS1} we present the estimation results for a randomly selected agent (here $k=13$) over a time horizon $t=5000$. The bandwidth parameter $h$ is set to $0.15$ for the whole experiment. Due to the collaboration with other agents, the quality of the results is satisfactory for the total region of interest and not only for the local domain of the considered agent. 

As we mentioned in the previous sections, transferring all the data to a single processing centre usually requires a significant communication cost. One of the main goals of distributed learning is to provide a result that is close to the centralised approach, and the proposed data exchange and aggregation algorithm has the possibility (under a proper number of connections between the nodes) to achieve it. In Fig.~\ref{fig:NS2}. we present a side-by-side comparison of the model provided by the single agent and the model calculated in a centralised way with the usage of all the agents' local data. 
\begin{figure}[H]
    \centering
\includegraphics[width=0.9\textwidth]{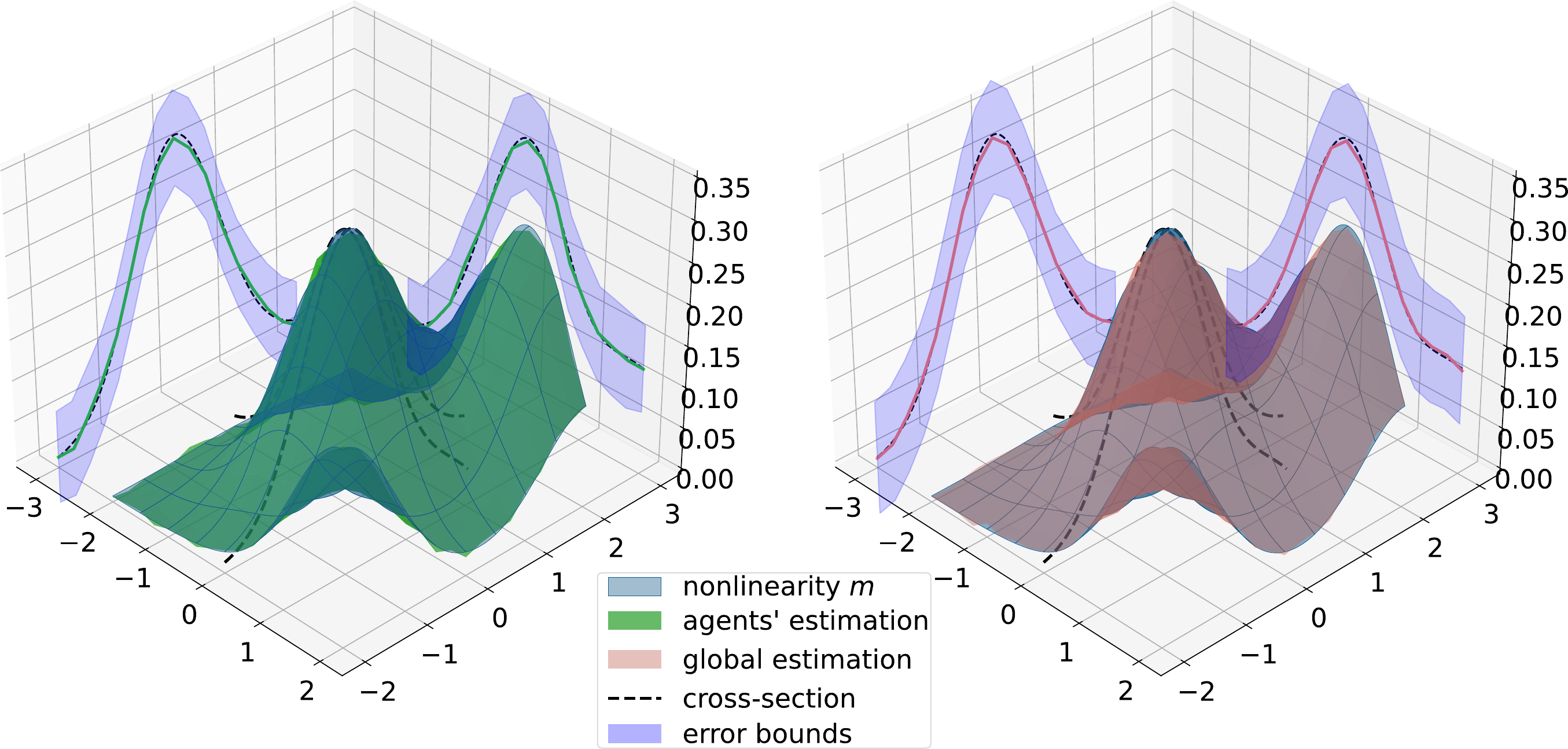}
    \caption{Comparison of the model provided by a single agent, with a global model, that uses all the agents' data.}
    \label{fig:NS2}
\end{figure}
In Fig.~\ref{fig:NS3} we present the evolution of our confidence bound over time for a selected estimation point $x=(0,0)$. At the beginning the bound is high due to the lack of reacquired tuples, but with time more tuples for the estimation point are obtained. This process however, slows down with time, since the number of agents in the network is finite, hence no new tuples are acquired and the improvement of the bounds is a result of updating existing tuples. 
\begin{figure}[H]
    \centering
\includegraphics[width=0.7\textwidth]{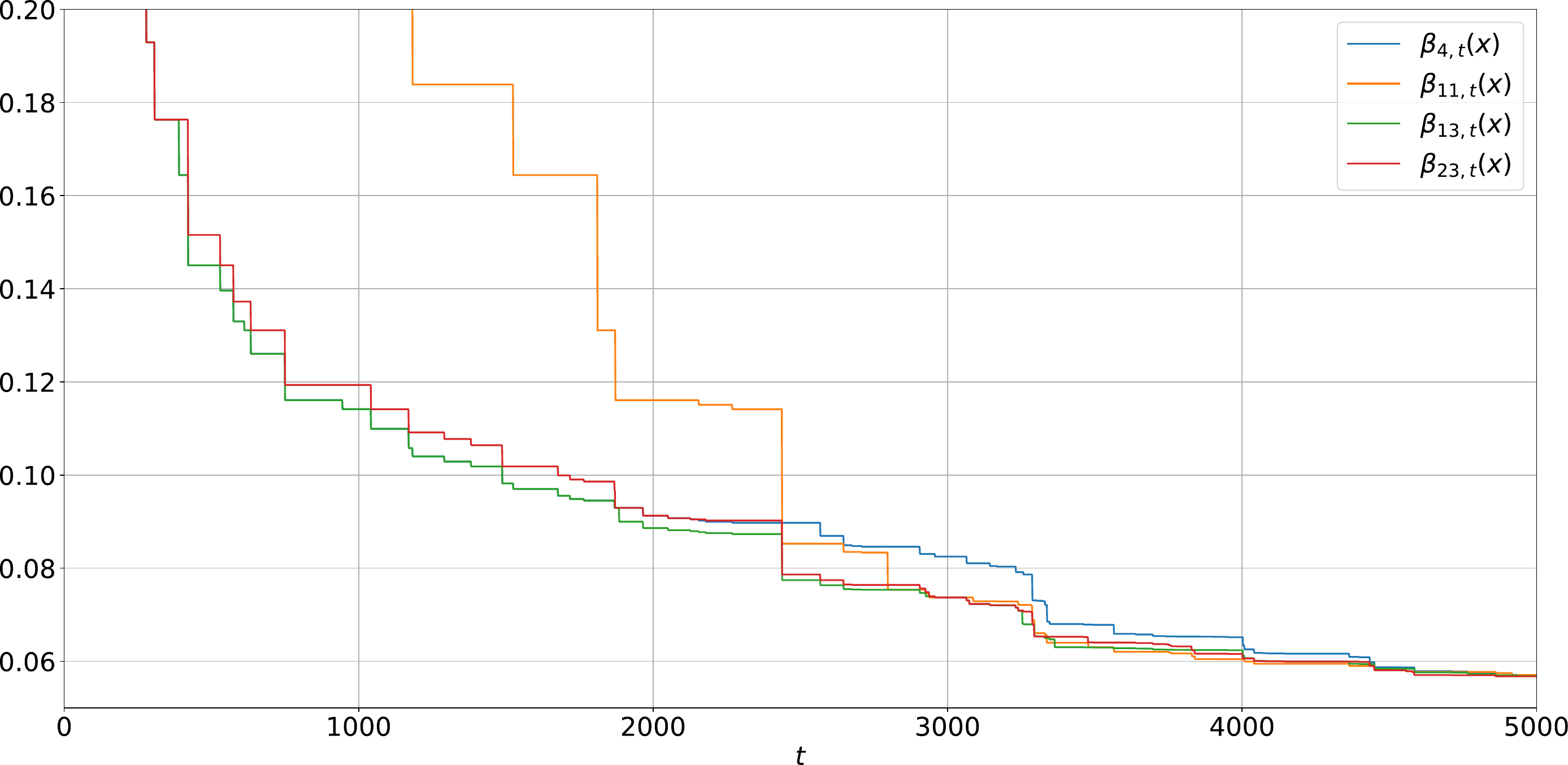}
    \caption{Bound evolution over time for a selected estimation point and a few selected agents. }
    \label{fig:NS3}
\end{figure}
\section{Conclusions}
In this paper, we have proposed a new distributed learning algorithm, designed for learning multivariate phenomena. Following the data exchange and aggregation procedure as described in Algorithm 1 and using a distributed estimator, a single agent is able to model a phenomenon in regions that are far beyond its local scope. We have also provided corresponding non-asymptotic error bounds that hold with a user-define confidence level.

We have formally investigated, under rather mild assumptions, the non-asymptotic properties of the proposed method. Also, we have illustrated the obtained theoretical results via numerical simulations, which clearly show the advantages of the techniques described in this paper. 

As future work, we plan to expand the proposed approach to a more general setting, where we allow the agents to have different bandwidth parameter $h$ and noise proxy variance upper bounds $\sigma$. Also we aim to investigate network topology properties in order to develop a technique for setting data sharing flags.

\section*{Appendix}
\begin{proof}[Proof \textnormal{(}of Lemma \ref{L:local}\textnormal{).}] We begin with the observation that
\begin{align} \label{eq:uppbA}
\norm{\sum_{n=1}^t \frac{K_h(x, \xi_{n})}{\kappa _{t}(x)} y_{n} - m(x)}_2 \leq \sum_{n=1}^{t}\theta_n \norm{ m( \xi _{n}) - m(x) }_2 + \norm{ \sum_{n=1}^t \theta_n \eta_{n} }_2,
\end{align}
where $\theta _n := K_h(x, \xi_{n}) / \kappa_{t}(x)$.
Note that $\Sigma_{n=1}^t \theta_n = 1$.
Due to Assumption $4$, if $K_h(x, \xi_{n}) > 0$, then $\norm{x - \xi _{n}}_2 / h \leq 1$. Therefore, (\textit{cf.}~Assumption $1$)
\begin{align*}
K_h(x, \xi_{n}) > 0\quad\Longrightarrow\quad \norm{ m(\xi_{n}) - m(x)}_2 \leq L \norm{ x -\xi_{n}}_2 \leq Lh,
\end{align*}
and since the weights $\theta _{n}$ sum up to $1$,
\begin{align*}
\sum_{n=1}^t \theta_n \norm{ m(\xi_{n}) - m(x) }_2 \leq Lh.
\end{align*}
For the last term in (\ref{eq:uppbA}), observe that
\begin{align}\label{eq:transf}
\norm{ \sum\nolimits_{n=1}^t \theta_n \eta_{n} }_2
= \frac{1}{\kappa_{t}(x)} \norm{ \sum\nolimits_{n=1}^t K_h(x, \xi_{n}) \eta _{n}}_2.
\end{align}
According to Lemma 2, the right-hand side of Eq.~\eqref{eq:transf} is upper bounded (with probability $1-\delta$) by
\begin{align*}
\frac{1}{\kappa_{t}(x)} \sigma \sqrt{2 \log \left( \delta ^{-1} \left(1 + \sum\nolimits_{n=1}^t K_h^2(x, \xi_{n})\right)^{d/2} \right) \left(1 + \sum\nolimits_{n=1}^t K_h^2 (x, \xi_{n})\right)}.
\end{align*}
Furthermore, since $K_h(x, \xi_{n}) \leq 1$ (\textit{cf.}~Assumption 4), we
obtain
\begin{align*}
\frac{1}{\kappa_{t}(x)} \norm{ \sum\nolimits_{n=1}^t K_h(x, \xi_{n}) \eta_{n}}_2 \leq \sigma \sqrt{2 \log\left( \delta^{-1} \left(1 + \kappa_{t}(x)\right)^{d/2}\right)}\frac{\sqrt{1 +\kappa_{t}(x)}}{\kappa_{t}(x)}.
\end{align*}
Observe next that, if $\kappa_{t}(x) > 1$, then
\begin{align*}
\frac{\sqrt{1 + \kappa_{t}(x)}}{\kappa_{t}(x)}
< \frac{\sqrt{2 \kappa_{t}(x)}}{\kappa_{t}(x)}
= \frac{\sqrt{2}}{\sqrt{\kappa_{t}(x)}}.
\end{align*}
Therefore, with probability $1-\delta$, for $\kappa_{t}(x) > 1$,
\begin{align*}
\frac{1}{\kappa_{k,t}(x)} \norm{ \sum\nolimits_{n=1}^t K_h(x, \xi_{n}) \eta_{n}}_2
\leq \frac{2 \sigma}{\kappa_{t}(x)} \sqrt{\kappa_{t}(x) \log \left( \delta^{-1} \left(1 + \kappa_{t}(x)\right)^{d/2}\right)},    
\end{align*}
whereas for $0 < \kappa_{t} \leq 1$,
\begin{align*}
\frac{1}{\kappa _{t}(x)} \left| \sum_{n=1}^t K_h(x, \xi_{n}) \eta_{n} \right|
&\leq \frac{\sigma}{\kappa_{t}(x)} \sqrt{2 \log\left(\delta^{-1} \left(1 + \kappa_{t}(x)\right)^{d/2}\right)} \sqrt{1 + \kappa_{t}(x)} \\
&\leq \frac{2 \sigma}{\kappa_{t}(x)} \sqrt{\log(\delta^{-1}2^{d/2})},
\end{align*}
which completes the proof.
\end{proof} 
   
\begin{lemma}\label{Lem:Tech_A}
 Let $\{v_t\in\mathbb{R}\colon\,t\in\mathbb{N}\}$ and $\{\eta_t\in\mathbb{R}^d\colon\,t\in\mathbb{N}\}$  be  stochastic processes. Assume that there exists some $\sigma>0$ such that, for every $\gamma_t\in\mathbb{R}^d$ (possibly a function of $v_t$), and every $t\in\mathbb{N}$,
    \begin{equation}
    \label{eqn:sub_gaussian}
        \mathbb{E}[\exp(\gamma_t^\top\eta_t)\vert \eta_{1:t-1},v_{1:t}]\le\exp\left(\frac{\gamma_t^\top\gamma_t\sigma^2}{2}\right).
    \end{equation}
Define
\begin{align*}
S_t := \sum\nolimits_{n=1}^t v_n \eta _n \quad \text{and} \quad V_t := \sum\nolimits_{n=1}^t v_n^2.
\end{align*}
Then, for every $t\in\mathbb{N}$ and $0<\delta <1$, with probability $1-\delta$,
\begin{align*}
S_t^\top S_t \leq 2 \sigma^2\log \left[ \frac{(V_t+1)^{d/2}}{\delta} \right] (V_t + 1).
\end{align*}
\end{lemma}

\begin{proof}
Without loss of generality, let $\sigma = 1$. For every $\lambda \in \mathbb{R}^d$ let
\begin{align*}
w_t(\lambda) :=\exp\left( \lambda^\top S_t - \frac{1}{2} \lambda^\top\lambda V_t \right).
\end{align*}
From Lemma 3, we note that, for every $\lambda \in \mathbb{R}^d$, $\mathbb{E}\{ w_t(\lambda)\} \leq 1$. Let now $\Lambda $ be a $\mathcal{N}( \textbf{0},\mathbb{I})$ random variable, independent of all other variables. Clearly, $\mathbb{E}\{ w_{t}(\Lambda) | \Lambda \} \leq 1$. Define
\begin{align*}
w_t := \mathbb{E}\{ w_t(\Lambda) | v_n, \eta_n\colon n\in\mathbb{N} \}.
\end{align*}
Then $\mathbb{E}\{ w_t \} \leq 1$ since, by the tower property of conditional expectation,
\begin{align*}
\mathbb{E}\{ w_t\}
&= \mathbb{E}\{\mathbb{E}\{ w_{t}( \Lambda ) | v_n, \eta_n\colon n\in\mathbb{N} \}  \} \\
&= \mathbb{E}\{ w_{t}( \Lambda )  \} \\
&= \mathbb{E}\{ \mathbb{E}\{ w_{t}(\Lambda) | \Lambda \} \} \\
&\leq 1.
\end{align*}
We can also express $w_t$ directly as
\begin{align*}
w_t 
&=\frac{1}{\sqrt{(2\pi)^d}} \int_{\mathbb{R}^d} \exp\left( \lambda^\top S_t - \frac{1}{2} \lambda^\top\lambda V_t \right) \exp\left( -\frac{1}{2}\lambda^\top\lambda\right) d\lambda \\
&= \frac{1}{\sqrt{(2\pi)^d}} \int_{\mathbb{R}^d} \exp\left( \lambda^\top S_t - \frac{1}{2} \lambda^\top\lambda V_t-\frac{1}{2}\lambda^\top\lambda \right)  d\lambda \\
&= \frac{1}{\sqrt{(2\pi)^d}} \int_{\mathbb{R}^d} \exp\left( \lambda^\top S_t - \frac{1}{2}(V_t+1)\lambda^\top\lambda \right)  d\lambda \\
&= \frac{1}{\sqrt{(2\pi)^d}} \int_{\mathbb{R}^d} \exp\left( \lambda^\top S_t - \frac{1}{2}(V_t+1)\lambda^\top\lambda +\frac{1}{2}\frac{1}{V_t+1}S_t^\top S_t-\frac{1}{2}\frac{1}{V_t+1}S_t^\top S_t\right)  d\lambda \\
&=\frac{1}{\sqrt{(2\pi)^d}} \int_{\mathbb{R}^d} \exp\left( -\frac{1}{2}(V_t+1)\left[\lambda-\frac{1}{V_t+1}S_t\right]^\top\left[\lambda-\frac{1}{V_t+1}S_t\right]\right)\exp\left(\frac{1}{2}\frac{1}{V_t+1}S_t^\top S_t \right)  d\lambda \\
&=(V_t + 1)^{-d/2} \exp\left(\frac{1}{2}\frac{1}{V_t+1}S_t^\top S_t \right).
\end{align*}
Therefore $\mathbb{P}\{ \delta w_t \geq 1 \} $ is equal to
\begin{align*}
&\mathbb{P}\left\{ \frac{\delta }{(V_t+1)^{d/2}}\exp\left( \frac{1}{2} \frac{S_t^\top S_t}{V_t + 1}\right) \geq 1  \right\} \\
&\qquad\qquad\qquad\qquad = \mathbb{P}\left\{  \exp\left( \frac{1}{2}\frac{S_t^\top S_t}{V_t + 1}\right) \geq (V_t+1)^{d/2} \frac{1}{\delta}  \right\} \\
&\qquad\qquad\qquad\qquad = \mathbb{P}\left\{ \frac{S_t^\top S_t}{V_t + 1} \geq 2 \log \left[ (V_t+1)^{d/2} \frac{1}{\delta} \right] \right\} \\
&\qquad\qquad\qquad\qquad = \mathbb{P}\left\{ S_t^\top S_t \geq 2 \log \left[ \frac{(V_t+1)^{d/2}}{\delta} \right] (V_t + 1) \right\}. 
\end{align*}
Recall now that $\mathbb{E}\{w_t  \}\leq 1$. Hence, due to Markov's inequality,
\begin{align*}
\mathbb{P}\{ \delta w_t \geq 1 \} \leq \delta\, \mathbb{E}\{ w_t \} \leq \delta,
\end{align*}
which completes the proof.
\end{proof}

\begin{lemma} 
    Let $\{v_t\in\mathbb{R}\colon\,t\in\mathbb{N}\}$ and $\{\eta_t\in\mathbb{R}^d\colon\,t\in\mathbb{N}\}$  be  stochastic processes. Assume that there exists some $\sigma>0$ such that, for every $\gamma_t\in\mathbb{R}^d$ (possibly a function of $v_t$), and every $t\in\mathbb{N}$
    \begin{equation}
    \label{eqn:sub_gaussian}
        \mathbb{E}[\exp(\gamma_t^\top\eta_t)\vert \eta_{1:t-1},v_{1:t}]\le\exp\left(\frac{\gamma_t^\top\gamma_t\sigma^2}{2}\right).
    \end{equation}
    Define, for every fixed $\lambda\in\mathbb{R}^d$, 
    \[
        \omega_t(\lambda)\coloneqq \exp\left(\sum_{n=1}^t\frac{\lambda^\top\eta_n v_n}{\sigma}-\frac{1}{2}\lambda^\top\lambda v_n^2\right).
    \]
    Then, $\mathbb{E}[\omega_t(\lambda)]\le 1$.
\end{lemma}
\begin{proof}
    Let
    \[
        D_n\coloneqq \exp\left(\frac{\lambda^\top\eta_n v_n}{\sigma}-\frac{1}{2}\lambda^\top\lambda v_n^2\right), \quad n = 1, \dots, t.
    \]
    Clearly, $\omega_t(\lambda)=D_1 D_2 \cdots D_t$.
    Note that
    \begin{align*}
        \mathbb{E}[D_t\mid \eta_{1:t-1},v_{1:t}]&=\mathbb{E}\left[\left. \exp\left(\frac{\lambda^\top\eta_t v_t}{\sigma}\right) \bigg/ \exp\left(\frac{1}{2}\lambda^\top\lambda v_t^2\right)\,\right|\, \eta_{1:t-1},v_{1:t}\right]\\
        &=\mathbb{E}\left[\exp\left(\left.\frac{\lambda^\top\eta_t v_t}{\sigma}\right)\right| \eta_{1:t-1},v_{1:t}\right] \bigg/ \exp\left(\frac{1}{2}\lambda^\top\lambda v_t^2\right).
    \end{align*}
    Let $\gamma_t = \frac{\lambda v_t}{\sigma}$, then $\exp\left(\frac{1}{2}\lambda^\top\lambda v_t^2\right)=\exp\left(\frac{\gamma_t^\top\gamma_t\sigma^2}{2}\right)$ and, due to~\eqref{eqn:sub_gaussian}, 
    \begin{align*}
    \mathbb{E}\left[\left. \exp\left(\frac{\lambda^\top\eta_t v_t}{\sigma}\right)\right| \eta_{1:t-1},v_{1:t}\right]\leq \exp\left(\frac{\gamma_t^\top\gamma_t\sigma^2}{2}\right), 
    \end{align*}
    therefore
    \[
        \mathbb{E}[D_t\mid \eta_{1:t-1},v_{1:t}] \le \frac{\exp\left(\frac{\gamma_t^\top\gamma_t\sigma^2}{2}\right)}{\exp\left(\frac{\gamma_t^\top\gamma_t\sigma^2}{2}\right)}=1.
    \]
    
    Next, for every $t\in\mathbb{N}$,
    \begin{align*}
        \mathbb{E}[\omega_t(\lambda)\mid\eta_{1:t-1},v_{1:t}]&=\mathbb{E}[D_1\cdots D_{t-1}D_t\mid \eta_{1:t-1},v_{1:t}]\\
        &= D_1\cdots D_{t-1}\mathbb{E}[D_t\mid \eta_{1:t-1},v_{1:t}]\leq\omega_{t-1}(\lambda).
    \end{align*}
    Therefore,
    \begin{align*}
        \mathbb{E}[\omega_t(\lambda)]&=\mathbb{E}[\mathbb{E}[\omega_t(\lambda)\mid\eta_{1:t-1},v_{1:t}]]\\
        &\le\mathbb{E}[\omega_{t-1}(\lambda)]\le\cdots\le \mathbb{E}[\omega_1(\lambda)]\\
        &=\mathbb{E}[\mathbb{E}[D_1\mid v_1]]\le 1,
    \end{align*}
    which completes the proof.
\end{proof}
\bibliographystyle{plain} 
\bibliography{Literature}

\begin{thebibliography}{10}

\bibitem{bertrand2011consensus}
Alexander Bertrand and Marc Moonen.
\newblock Consensus-based distributed total least squares estimation in ad hoc wireless sensor networks.
\newblock {\em IEEE Transactions on Signal Processing}, 59(5):2320--2330, 2011.

\bibitem{cattivelli2010distributed}
Federico~S. Cattivelli and Ali~H. Sayed.
\newblock Distributed nonlinear kalman filtering with applications to wireless localization.
\newblock In {\em 2010 IEEE International Conference on Acoustics, Speech and Signal Processing}, pages 3522--3525. IEEE, 2010.

\bibitem{jkedrzejewski2022purchasing}
Arkadiusz Jedrzejewski, Katarzyna Sznajd-Weron, Jakub Paw{\l}owski, and Anna Kowalska-Pyzalska.
\newblock Purchasing decisions on alternative fuel vehicles within an agent-based model.
\newblock In {\em International Conference on Computational Science}, pages 719--726. Springer, 2022.

\bibitem{lagosz2021identification}
Szymon {\L}agosz, Przemyslaw {\'S}liwi{\'n}ski, and Pawel Wachel.
\newblock Identification of wiener--hammerstein systems by $\ell$1--constrained volterra series.
\newblock {\em European Journal of Control}, 58:53--59, 2021.

\bibitem{lopes2008diffusion}
Cassio~G. Lopes and Ali~H. Sayed.
\newblock Diffusion least-mean squares over adaptive networks: Formulation and performance analysis.
\newblock {\em IEEE Transactions on Signal Processing}, 56(7):3122--3136, 2008.

\bibitem{mahela2020comprehensive}
Om~Prakash Mahela, Mahdi Khosravy, Neeraj Gupta, Baseem Khan, Hassan~Haes Alhelou, Rajendra Mahla, Nilesh Patel, and Pierluigi Siano.
\newblock Comprehensive overview of multi-agent systems for controlling smart grids.
\newblock {\em CSEE Journal of Power and Energy Systems}, 8(1):115--131, 2020.

\bibitem{malecki2023multi}
Krzysztof Ma{\l}ecki, Patryk G{\'o}rka, and Maria Gokieli.
\newblock Multi-agent cellular automaton model for traffic flow considering the heterogeneity of human delay and accelerations.
\newblock In {\em International Conference on Computational Science}, pages 539--552. Springer, 2023.

\bibitem{modalavalasa2021review}
Sowjanya Modalavalasa, Upendra~Kumar Sahoo, Ajit~Kumar Sahoo, and Satyakam Baraha.
\newblock A review of robust distributed estimation strategies over wireless sensor networks.
\newblock {\em Signal Processing}, 188:108150, 2021.

\bibitem{rasheed2022review}
Ammar Abdul~Ameer Rasheed, Mohammed~Najm Abdullah, and Ahmed~Sabah Al-Araji.
\newblock A review of multi-agent mobile robot systems applications.
\newblock {\em International Journal of Electrical \& Computer Engineering (2088-8708)}, 12(4), 2022.

\bibitem{sakurada2020multi}
Lucas Sakurada and Paulo Leit{\~a}o.
\newblock Multi-agent systems to implement industry 4.0 components.
\newblock In {\em 2020 IEEE Conference on Industrial Cyberphysical Systems (ICPS)}, volume~1, pages 21--26. IEEE, 2020.

\bibitem{sim2011agent}
Kwang~Mong Sim.
\newblock Agent-based cloud computing.
\newblock {\em IEEE Transactions on Services Computing}, 5(4):564--577, 2011.

\bibitem{jain2010innovations}
Dipti Srinivasan and Lakhmi~C. Jain.
\newblock {\em Innovations in Multi-Agent Systems and Applications-1}.
\newblock Springer, 2010.

\bibitem{wachel2023learning}
Pawel Wachel, Koen Tiels, and Maciej Fili{\'n}ski.
\newblock Learning low-dimensional separable decompositions of mimo non-linear systems.
\newblock {\em International Journal of Control}, 96(4):900--906, 2023.

\bibitem{WACHEL2023100912}
Paweł Wachel, Krzysztof Kowalczyk, and Cristian~R. Rojas.
\newblock Decentralized diffusion-based learning under non-parametric limited prior knowledge.
\newblock {\em European Journal of Control}, 75:100912, 2024.

\end{thebibliography}
\end{document}